\documentclass[10pt, conference]{IEEEtran}
\usepackage{epsf,psfrag,amssymb,amsfonts,cite}
\ifCLASSOPTIONcompsoc
\usepackage[caption=false,font=normalsize,labelfon
t=sf,textfont=sf]{subfig}
\else
\usepackage[caption=false,font=footnotesize]{subfi
g}
\fi
\newtheorem{theorem}{Theorem}
\newtheorem{example}{Example}
\newtheorem{corollary}{Corollary}

\newtheorem{definition}{Definition}

\def\psfancypar#1#2{\begingroup\def\par{\endgraf\endgroup\lineskiplimit=0pt}
               \setbox2=\hbox{\large\sc #2}
               \newdimen\tmpht \tmpht \ht2 \advance\tmpht by \baselineskip
               \font\hhuge=Times-Bold at \tmpht
               \setbox1=\hbox{{\hhuge #1}}
               \count7=\tmpht \count8=\ht1
               \divide\count8 by 1000 \divide\count7 by \count8 
               \tmpht=.001\tmpht\multiply\tmpht by \count7 
               \font\hhuge=Times-Bold at \tmpht
               \setbox1=\hbox{{\hhuge #1}}
               \noindent
                \hangindent1.05\wd1
               \hangafter=-2 {\hskip-\hangindent
               \lower1\ht1\hbox{\raise1.0\ht2\copy1}%
                \kern-0\wd1}\copy2\lineskiplimit=-1000pt}

\newcommand{\beq}{\begin{equation}}
\newcommand{\eeq}{\end{equation}}
\newcommand{\bqa}{\begin{eqnarray}}
\newcommand{\eqa}{\end{eqnarray}}
\newcommand{\bqn}{\begin{eqnarray*}}
\newcommand{\eqn}{\end{eqnarray*}}

\newcommand{\be}{\begin{enumerate}}
\newcommand{\ee}{\end{enumerate}}
\newcommand{\bi}{\begin{itemize}}
\newcommand{\ei}{\end{itemize}}
\newcommand{\bd}{\begin{description}}
\newcommand{\ed}{\end{description}}
\newcommand{\ba}{\begin{array}}
\newcommand{\ea}{\end{array}}
\newcommand{\bde}{\begin{definition}}
\newcommand{\ede}{\end{definition}}
\newcommand{\bex}{\begin{example}}
\newcommand{\eex}{\end{example}}

\def\boxit#1{\vbox{\hrule\hbox{\vrule\kern3pt
        \vbox{\kern3pt#1\kern3pt}\kern3pt\vrule}\hrule}}

\def\reals{ { {\rm  I \kern-0.15em R }  } }
\def\complex{ {\,{{\rm C} \kern-0.50em \raise0.20ex {  |}}\, }}

\def\0bf{{\bf 0}}
\def\1bf{{\bf 1}}
\def\2bf{{\bf 2}}
\def\3bf{{\bf 3}}
\def\4bf{{\bf 4}}
\def\5bf{{\bf 5}}
\def\6bf{{\bf 6}}
\def\7bf{{\bf 7}}
\def\8bf{{\bf 8}}
\def\9bf{{\bf 9}}

\def\Rbf{{\bf R}}

\def\Rxx{\Rbf_{\ssstyle X\kern-.1em X}}

\let\ssstyle=\scriptscriptstyle

\def\Kout{\setbox1=\hbox{\Huge\bf K}\hbox to
1.05\wd1{\hspace{.05\wd1}
}}
\def\Sout{\setbox1=\hbox{\Huge\bf S}\hbox to 1.05\wd1{\hspace{.05\wd1}
}}

\title{\LARGE The Han-Kobayashi Region for a Class of Gaussian Interference Channels with Mixed Interference}

\author{{Yu Zhao, Fangfang Zhu and Biao Chen} \\
Department of EECS \\
Syracuse University \\
yzhao05\{fazhu, bichen\}@syr.edu
}

\begin{document}
\maketitle
\begin{abstract}
A simple encoding scheme based on Sato's non-na\"{i}ve frequency division is proposed for a class of Gaussian interference channels with mixed interference. The achievable region is shown to be equivalent to that of Costa's {\em noiseberg} region for the one-sided Gaussian interference channel. This allows for an indirect proof that this simple achievable rate region is indeed equivalent to the Han-Kobayashi (HK) region with Gaussian input and with time sharing for this class of Gaussian interference channels with mixed interference.
\end{abstract}

\section{Introduction}
The interference channel (IC) describes a network where multiple transmitters communicate with their intended receivers via a common medium. The characterization of the capacity region for a two-user IC is an open problem except for the strong and very strong interference cases \cite{Carleial:75IT,Sato:81IT,Han&Kobayashi:81IT}. To date, the largest achievable rate region is the celebrated Han-Kobayashi (HK) region that employs rate splitting at the transmitters and simultaneous decoding at the receivers \cite{Han&Kobayashi:81IT}. Not surprisingly, for those ICs whose capacity regions are completely characterized, it is without an exception that the capacity region coincides with the HK region.

However, the general HK region involves a time sharing variable that makes its evaluation intractable. For the Gaussian interference channel (GIC), another difficulty is the input distribution. A two-user GIC in its standard form can be represented as
\beq
\begin{array} {ccc}
Y_1&=&X_1+bX_2+Z_1, \\
Y_2&=&aX_1+X_2+Z_2,
\end{array}
\label{eq:GIC}
\eeq
where $X_1$ and $X_2$ are the input signals and are subject to respective power constraints $P_1$ and $P_2$; $Y_1$ and $Y_2$ are the received signals; $Z_1$ and $Z_2$ are Gaussian noises of unit variance and are independent of the inputs $X_1$ and $X_2$. This model is depicted in Fig.~\ref{Fig:f1}. While for all the cases where the capacity results are known for a GIC, the optimal input distribution is invariably Gaussian, it is not yet known (or proven) that such is the case for the general GIC. 

There has been recent progress in obtaining computable subregion of the HK achievable region using Sato's non-na\"{i}ve frequency division \cite{Shang&Chen:07ISIT}. For the one-sided GIC (denoted as ZGIC) shown in Fig.~\ref{Fig:5}(a),  Motahari and Khandani established that such a non-na\"{i}ve frequency division scheme achieves the HK region with Gaussian input \cite{Motahari&Khandani:09IT}. Most recently, Costa introduced the so-called {\em noiseberg} scheme which uses water filling to achieve optimal power sharing between two orthogonal dimensions \cite{Costa:Noiseberg}. It turns out, as shown in the next section, that this simple noiseberg scheme achieves precisely the same HK region with Gaussian input. 

\begin{figure}
\centering
\includegraphics{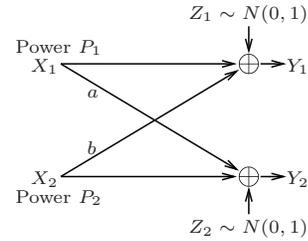}
\caption{ Two-user GIC}\label{Fig:f1}
\end{figure}

This paper focuses on GICs with mixed interference (MGIC) and with $ab\geq 1$, $a\leq 1$ and $b\geq 1$ (cf. Fig.~\ref{Fig:f1} and Eq.~(\ref{eq:GIC})). We describe a simple coding scheme that combines the noiseberg scheme with that of simultaneous decoding at the receiver experiencing strong interference. The obtained rate region is subsequently shown to coincide with the HK region with Gaussian input.


The rest of the paper is organized as follows. In Section \ref{sec:ZGIC}, we review the noiseberg scheme for the ZGIC and provide a proof of its equivalence to the HK region with Gaussian input. Section \ref{MGIC} describes the coding scheme for a class of MGIC and proves that the scheme achieves the HK region with Gaussian input. Section \ref{Conclusion} concludes this paper.

\section{Noiseberg Region for the ZGIC}\label{sec:ZGIC} 

We consider the degraded GIC shown in Fig.~\ref{Fig:5}(b), which is equivalent to the ZGIC with $a<1$ in Fig.~\ref{Fig:5}(a)\cite{Costa:85IT}.

The noiseberg region, denoted by $\mathcal{R_{N}}$ and introduced by Costa in \cite{Costa:Noiseberg} for a ZGIC with weak interference ($a<1$ in Fig.~\ref{Fig:5}(a)) is the set of all nonnegative rate pairs $(R_1, R_2)$ satisfying
\small
\bqa
R_{1} &\leq& \bar{\lambda}R_{1 \bar{\lambda}} + \lambda R_{1 \lambda}, \nonumber \\
R_{2}  &\leq& \bar{\lambda}R_{2 \bar{\lambda}}, \nonumber 
\eqa
\normalsize
where
\small
\begin{eqnarray}
R_{1 \bar{\lambda}} &\leq& \frac{1}{2} \log \left(1+ \frac{P_{1A}}{\bar{\lambda}}\right) +\frac{1}{2} \log \left(1+ \frac{a^{2} \frac{P_{1C}}{\bar{\lambda}}}{1+a^{2}\frac{P_{1A}}{\bar{\lambda}}+ \frac{P_{2}}{\bar{\lambda}}}\right), \nonumber \\\label{eq:2}\\
R_{2 \bar{\lambda}}&\leq& \frac{1}{2} \log \left(1+ \frac{\frac{P_{2}}{\bar{\lambda}}}{1+a^{2}\frac{P_{1A}}{\bar{\lambda}}}\right), \label{eq:3}\\
R_{1 \lambda} &\leq& \frac{1}{2} \log \left(1+ \frac{P_{1B}}{\lambda}\right) \label{eq:1},
\end{eqnarray}
\normalsize
and the power limits $P_{1A}$, $P_{1B}$ and $P_{1C}$ are determined by two parameters $h$ and $\lambda$ such that
\small
\begin{eqnarray}
\frac{P_{1A}}{\bar{\lambda}} &=& P_{1}-\frac{P_{2}\lambda}{a^{2}\bar{\lambda}}- \lambda \min \left\{h, \frac{1-a^2}{a^2}\right\} \nonumber \\ &&-\max\left \{0,h - \frac{1-a^2}{a^2}\right\}, \nonumber \\
\frac{P_{1B}}{\lambda} &=& P_{1}+\frac{P_{2}}{a^{2}}+ \bar{\lambda} \min\left\{h, \frac{1-a^{2}}{a^{2}}\right\},\nonumber \\
\frac{P_{1C}}{\bar{\lambda}} &=& \max\left\{0, h-\frac{1-a^{2}}{a^{2}}\right\}. \nonumber 
\end{eqnarray}
\normalsize

\begin{figure}[!t]
\centering
\subfloat[ZGIC]{\includegraphics{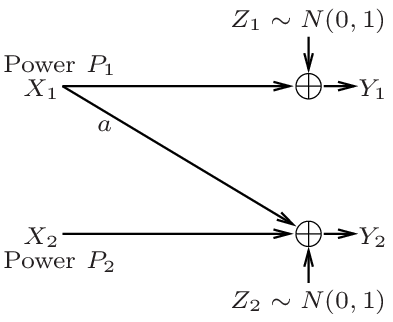}}
\subfloat[Degraded GIC ]{\includegraphics{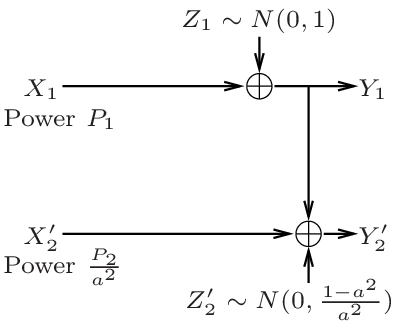}}
\caption{ZGIC}
\label{Fig:5}
\end{figure}

Costa \cite{Costa:Noiseberg} showed that $\mathcal{R_{N}}$ is achievable for the ZGIC with weak interference by a coding scheme that uses a two-band non-na\"{i}ve frequency division multiplexing (FDM) with water filling for optimal power allocation between the two subbands. The coding scheme, as well as its achievable region, involves two parameters $0\leq \lambda\leq 1$ and $h\geq 0$. They vary over the admissible region as shown in Fig. \ref{Fig:f2}, resulting in different transmission schemes depending on the values of the parameters. The parameter $\lambda$ determines how to allocate the frequency band.

\begin{itemize}
\item
The $\lambda$ subband is reserved for the communication between transmitter 1 and receiver 1.
\item 
The $\bar{\lambda}$ subband is shared by both transceiver pairs and the corresponding coding scheme is determined by the other parameter $h$.
\end{itemize}

As the noise $Z_2$ does not affect the transmission of $X_1$, water filling allows the overall power level in the $\lambda$-subband to be raised above that of the $\bar{\lambda}$-subband, with part of the noise spectrum of $Z_2$ floating above the signal level. This phenomenon, i.e., the existence of difference in heights of power spectrum for the two subbands is referred to as the \emph{noiseberg}. The parameter h is defined as the height of total power density in the $\lambda$-subband above that of $X_2$'s power density in the $\bar{\lambda}$-subband. Different $h$ values divide the admissible region for the parameter pairs into two regions, each employing a different coding scheme in the $\bar{\lambda}$-subband:


\noindent {\bf Multiplex region} This corresponds to $h \leq \frac{1-a^2}{a^2}$. As shown in Fig.~\ref{Fig:f3}, $Z_2'$ prevents user 1's power from spilling over to the $\bar{\lambda}$-band thus no rate-splitting is involved. Receiver 2 decodes $W_1$ first, subtracts it and decodes $W_2$.

\noindent  {\bf Overflow region} This corresponds to $h > \frac{1-a^2}{a^2}$. As shown in Fig.~\ref{Fig:f4}, water-filling of user 1's power occurs as the power spills over from the $\lambda$-subband to the $\bar{\lambda}$-subband. The encoding scheme in the $\bar{\lambda}$ subband thus involves rate splitting for $W_1$: a common message $W_{1c}$ with power $P_{1c}$ decoded by both receivers and a private message $W_{1p}$ with power $P_{1A}$ decoded only by receiver 1. Receiver 2 decodes $W_{1c}$ first, subtracts it, and decodes $W_2$, all the while treating $W_{1p}$ as noise.

\begin{figure}
\centering
\includegraphics{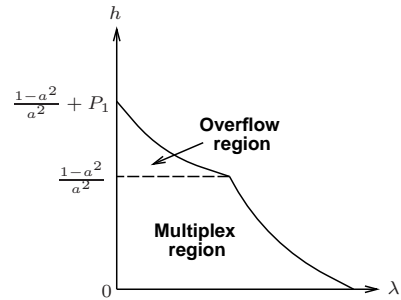}
\caption{Admissible region for $(\lambda, h)$}\label{Fig:f2}
\end{figure}
\begin{figure}
\centering
\includegraphics{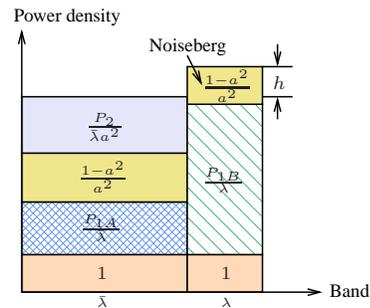}
\caption{Multiplex region}\label{Fig:f3}
\end{figure}
\begin{figure}
\centering
\includegraphics{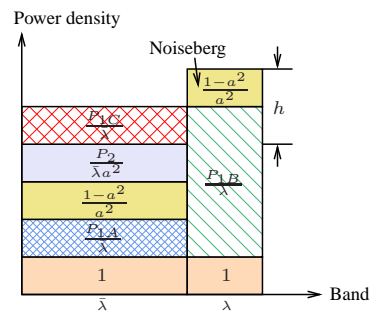}
\caption{Overflow region}\label{Fig:f4}
\end{figure}

It is remarkable that such a simple transmission scheme turns out to achieve precisely the same HK region with time sharing and with Gaussian input.
\begin{theorem} \label{lemma:Noiseberg}
For the weak ZGIC, $\mathcal{R_{N}} = \mathcal{R_{HK}}$.
\end{theorem}
\begin{proof}
Motahari and Khandani showed that for the ZGIC, the non-na\"{i}ve FDM region, denoted by $\mathcal{R_{FDM}}$, is equivalent to $\mathcal{R_{HK}}$, whose boundary points can be characterized by the optimization problem \cite[Eq. (151)]{Motahari&Khandani:09IT}. It suffices to verify the equivalence between $\mathcal{R_{N}}$ and $\mathcal{R_{FDM}}$.

We start by considering water filling in the two-band FDM applied to the degraded GIC shown in Fig.~\ref{Fig:5}(b). First, we split $W_1$ into private message $W_{1p}$ with power constraint $P_{1p}$ and common message $W_{1c}$ with power constraint $P_{1c}$ such that $P_{1p}+P_{1c}=P_1$. Power allocation into $\lambda$ and $\bar{\lambda}$ subbands is done in the following order. First, $P_{1p}$ is allocated to the two subbands in an arbitrary way. On top of that, $P_2$ is allocated to the two subbands via waterfilling. As $Y_2'$ sees additional noise $Z_2'$, $P_2$ is allocated on top of $Z_2'$ (see, e.g., Fig.~\ref{Fig:6}(d)). Finally, $P_{1c}$ is allocated to the two subbands, again, using waterfilling.
\begin{figure*}[!t]
\centering
\subfloat[]{\label{Fig:6a}\includegraphics{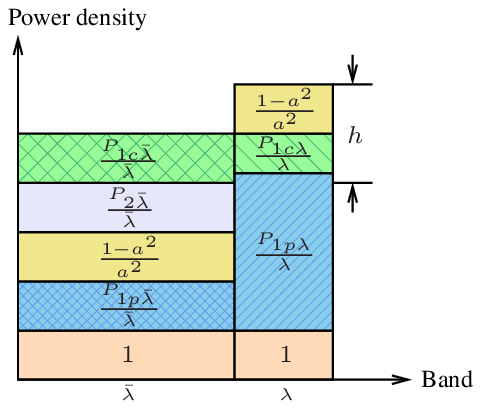}}
\subfloat[]{\label{Fig:6b}\includegraphics{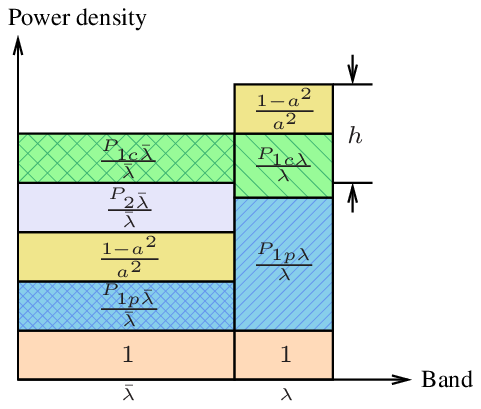}}
\\\subfloat[]{\label{Fig:6c}\includegraphics{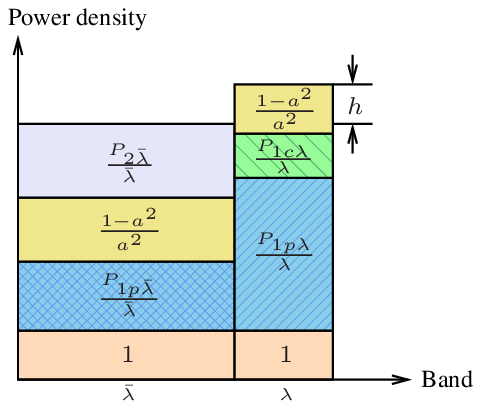}}
\subfloat[]{\label{Fig:6d}\includegraphics{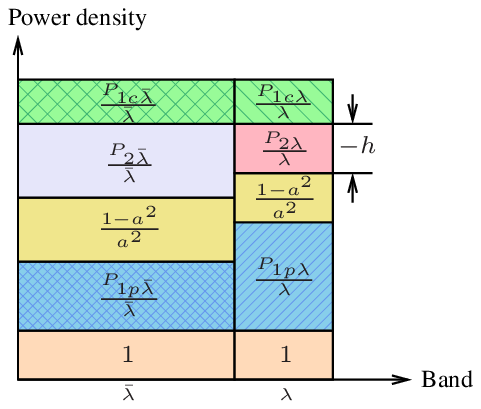}}
\caption{Possible power allocation outcomes of the two-band non-na\"{i}ve FDM scheme with ZGIC}
\label{Fig:6}
\end{figure*}

Depending on $P_{1p}$ and its allocation between the two subbands, there are four possible power allocation outcomes of this scheme, as shown in Fig.~\ref{Fig:6}. Since the scenarios illustrated in Fig.~\ref{Fig:6}(a)(b)(c) are equivalent to noiseberg cases, it remains to argue that the power allocation scheme with flat spectrum top as in Fig.~\ref{Fig:6}(d) is not optimal. This is because the achievable rates under such a scheme are formed by convex combinations of points on the curve of associated broadcast channel capacity, as the flat top requires $\frac{P_{1c \bar{\lambda}}}{\bar{\lambda}} = \frac{P_{1c \lambda}}{\lambda}$. As the broadcast channel capacity curve is convex, we can only achieve the points on the chord, which are dominated by the points on the curve corresponding to the scheme with no frequency division. Thus they are not optimal.

Next we generalize this conclusion to three-band FDM. Similarly we argue that the power-bandwidth allocation schemes with two or more subbands sharing the same flat top are not optimal. If two subbands are occupied solely by the transmission of $X_{1}$, they can be merged into one and this reduces to one of the noiseberg scenarios.
\end{proof}

\section{Achievable Region for a Class of MGIC} \label{MGIC}
For the MGICs with $0<a \leq 1$, $b \geq 1$ and $ab \geq 1$, the HK region with Gaussian input can be simplified to be the set of all rate pairs $(R_1, R_2)$ satisfying
\small
\bqa
R_1 &\leq& \sum\limits_{i=1}^{q} \lambda_i \left\{\frac{1}{2} \log (1 + P_{1i})\right\},\nonumber \\
R_2 &\leq&  \sum\limits_{i=1}^{q} \lambda_i \left\{\frac{1}{2} \log \left(1 + \frac{P_{2i}}{1+a^2 \alpha_i P_{1i}}\right)\right\}, \nonumber \\
R_1 + R_2 &\leq&  \sum\limits_{i=1}^{q} \lambda_i \left\{\frac{1}{2} \log \left(1+ \frac{P_{2i}+a^2\bar{\alpha}_i P_{1i}}{1+a^2\alpha_i P_{1i}}\right) \right. \nonumber \\
&&\left. + \frac{1}{2} \log (1+\alpha_i P_{1i})\right\},\nonumber
\eqa
\normalsize
where $q \in \mathbb{N}$, $i \in \{1, \dots, q\}$, $0 \leq \alpha_i \leq 1$, , $\bar{\alpha}_i = 1-\alpha_i$, $\lambda_i \geq 0$, $\sum\limits_{i=1}^{q} \lambda_i = 1$, $\sum\limits_{i=1}^{q} \lambda_i P_{1i} \leq P_1$ and $\sum\limits_{i=1}^{q} \lambda_i P_{2i} \leq P_2$. We refer to $\mathcal{R_{HK}}$ as the above HK region with Gaussian input. The difficulty in using this region is largely due to the presence of the time sharing variable.

We now describe a simple transmission scheme for a MGIC with $ab\geq 1$. The scheme resembles the noiseberg scheme as it also utilizes the two-band non-na\"{i}ve FDM. Specifically,
in the $\lambda$-subband, only transmitter 1 transmits and receiver 1 decodes $W_1$ with a rate constraint defined in (\ref{eq:1}).

In the $\bar{\lambda}$-subband, transmitter 2 does not use rate splitting. Transmitter 1, on the other hand, employs two encoding schemes depending on the value of $h$. The corresponding decoding schemes are also different. We describe them in details below.

\noindent{\bf Multiplex region} This corresponds to $h \leq \frac{1-a^2}{a^2}$. Sequential decoding is used at both receivers. Receiver 1 first decodes $W_2$, which leads to the constraint
\small
\bqa
R_{2\bar{\lambda}} & \le & \frac{1}{2} \log \left(1+ \frac{b^{2}P_{2}/\bar{\lambda}}{1+ P_{1A}/\bar{\lambda}}\right). \label{eq:4}
\eqa
\normalsize
Subsequently, $X_2$ is subtracted from $Y_1$ and $W_1$ is decoded with constraint (\ref{eq:2}). Receiver 2 decodes $W_2$ with constraint (\ref{eq:3}). Comparing (\ref{eq:3}) and (\ref{eq:4}), we see that (\ref{eq:4}) is redundant. Therefore sequential decoding in the multiplex region achieves $\mathcal{R_{N}}$.

\noindent{\bf Overflow region} This corresponds to $h > \frac{1-a^2}{a^2}$. Receiver 1 employs simultaneous decoding while receiver 2 still uses sequential decoding. For receiver 1, let $S_{1}$, $T_{1}$ and $T_{2}$ be, respectively, the rates of transmitter 1's private message $W_{1p}$, common message $W_{1c}$ and $W_2$. Then $R_{1\bar{\lambda}} = S_{1}+T_{1}$ and $R_{2\bar{\lambda}} = T_{2}$. Evaluation of error probability will give us
\small
\begin{eqnarray}
S_{1} &\leq& \frac{1}{2}\log \left(1+\frac{P_{1A}}{\bar{\lambda}}\right), \label{eq:5}\\
T_{1} &\leq& \frac{1}{2}\log \left(1+\frac{P_{1C}}{\bar{\lambda}}\right),\\
T_{2} &\leq& \frac{1}{2}\log \left(1+\frac{b^{2}P_2}{\bar{\lambda}}\right),\\
S_{1}+T_{1} &\leq& \frac{1}{2}\log \left(1+\frac{P_{1A}}{\bar{\lambda}}+\frac{P_{1C}}{\bar{\lambda}}\right),\\
S_{1}+T_{2} &\leq& \frac{1}{2}\log \left(1+\frac{P_{1A}}{\bar{\lambda}}+\frac{b^{2}P_2}{\bar{\lambda}}\right),\\
T_{1}+T_{2} &\leq& \frac{1}{2}\log \left(1+\frac{P_{1C}}{\bar{\lambda}}+\frac{b^{2}P_2}{\bar{\lambda}}\right),\\
S_{1}+T_{1}+T_{2} &\leq& \frac{1}{2}\log \left(1+\frac{P_{1A}}{\bar{\lambda}}+\frac{P_{1C}}{\bar{\lambda}}+\frac{b^{2}P_2}{\bar{\lambda}}\right).
\end{eqnarray}
\normalsize

As receiver 2 decodes $W_{1C}$ and $W_{2}$ sequentially, there will be two constraints
\small
\begin{eqnarray}
T_{1} &\leq& \frac{1}{2} \log \left(1+ \frac{a^{2} \frac{P_{1C}}{\bar{\lambda}}}{1+a^{2}\frac{P_{1A}}{\bar{\lambda}}+ \frac{P_2}{\bar{\lambda}}}\right), \\
T_{2} &\leq& \frac{1}{2} \log \left(1+ \frac{ \frac{P_2}{\bar{\lambda}}}{1+a^{2}\frac{P_{1A}}{\bar{\lambda}}}\right). \label{eq:6}
\end{eqnarray}
\normalsize

Fourier-Motzkin elimination on (\ref{eq:5})-(\ref{eq:6}) gives us
\small
\begin{eqnarray}
R_{1\bar{\lambda}} &\leq& \frac{1}{2}\log \left(1+\frac{P_{1A}}{\bar{\lambda}}+\frac{P_{1C}}{\bar{\lambda}}\right), \label{eq:7}\\
R_{1\bar{\lambda}} &\leq& \frac{1}{2}\log \left(1+\frac{P_{1A}}{\bar{\lambda}}\right) \nonumber \\ &&+  \frac{1}{2} \log \left(1+ \frac{a^{2} \frac{P_{1C}}{\bar{\lambda}}}{1+a^{2}\frac{P_{1A}}{\bar{\lambda}}+ \frac{P_2}{\bar{\lambda}}}\right), \label{eq:9}\\
R_{2\bar{\lambda}} &\leq& \frac{1}{2} \log \left(1+ \frac{ \frac{P_2}{\bar{\lambda}}}{1+a^{2}\frac{P_{1A}}{\bar{\lambda}}}\right), \label{eq:10}\\
R_{1\bar{\lambda}}+R_{2\bar{\lambda}} &\leq& \frac{1}{2}\log \left(1+\frac{P_{1A}}{\bar{\lambda}}+\frac{P_{1C}}{\bar{\lambda}}+\frac{b^{2}P_2}{\bar{\lambda}}\right), \label{eq:11}\\
R_{1\bar{\lambda}}+R_{2\bar{\lambda}} &\leq& \frac{1}{2}\log \left(1+\frac{P_{1A}}{\bar{\lambda}}+\frac{b^{2}P_2}{\bar{\lambda}}\right) \nonumber \\ &&+  \frac{1}{2} \log \left(1+ \frac{a^{2} \frac{P_{1C}}{\bar{\lambda}}}{1+a^{2}\frac{P_{1A}}{\bar{\lambda}}+ \frac{P_2}{\bar{\lambda}}}\right), \label{eq:12}\\
R_{1\bar{\lambda}} + 2R_{2\bar{\lambda}} &\leq& \frac{1}{2}\log \left(1+\frac{P_{1A}}{\bar{\lambda}}+\frac{b^{2}P_2}{\bar{\lambda}}\right) \nonumber \\&&+ \frac{1}{2}\log \left(1+\frac{P_{1C}}{\bar{\lambda}}+\frac{b^{2}P_2}{\bar{\lambda}}\right). \label{eq:8}
\end{eqnarray}
\normalsize

Then the achievable rate region in overflow region is
\small
\begin{eqnarray}
R_{1} &\leq& \bar{\lambda}R_{1 \bar{\lambda}} + \lambda R_{1 \lambda}, \nonumber \\
R_{2}  &\leq& \bar{\lambda}R_{2 \bar{\lambda}}, \nonumber
\end{eqnarray}
\normalsize
where $R_{1 \bar{\lambda}}$, $R_{2 \bar{\lambda}}$ and $R_{1 \lambda}$ satisfy (\ref{eq:7})-(\ref{eq:8}) and (\ref{eq:1}). In the appendix we simplify this region and prove that it is equivalent to $\mathcal{R_{N}}$.

Combining the results in both the multiplex and overflow regions, we conclude that $\mathcal{R_{N}}$ is achievable for the MGICs with $ab>1$. Thus we have,

\begin{theorem} \label{theorem: Mixed IC}
For the MGICs with $ab>1$, $\mathcal{R_{N}}$ associated with the ZGIC obtained by removing the interfering link with gain $b$ is achievable.
\end{theorem}

\begin{corollary}
For the MGICs with $ab>1$, $\mathcal{R_{N}} = \mathcal{R_{HK}}$.
\end{corollary}
\begin{proof}
On the one hand, with Gaussian input, $\mathcal{R_{HK}}$ of the MGIC is a subset of that of the ZGIC. On the other hand, $\mathcal{R_{N}} \subseteq \mathcal{R_{HK}}$ for the MGIC in general. But Theorem 1 states that $\mathcal{R_{N}} = \mathcal{R_{HK}}$ for the ZGIC. This establishes the equivalence between $\mathcal{R_{N}}$ and $\mathcal{R_{HK}}$ for the MGICs with $ab>1$. 
\end{proof}

This is consistent with \cite[Theorem 12]{Motahari&Khandani:09IT} that states for the MGICs with $ab>1$, $\mathcal{R_{HK}}$ is equivalent to that of the corresponding ZGIC.
\section{Conclusion} \label{Conclusion}
This paper established the equivalence between Costa's noiseberg region and the HK region with Gaussian input for the ZGIC. For the MGICs with $ab>1$, an achievable rate region was developed by rate splitting and a mixture of sequential and simultaneous decoding. By comparing the new region to the noiseberg region of the ZGIC, we proved that it is a simplification of the HK region for the MGICs with $ab>1$ and with Gaussian codebook.

\section*{Acknowledgment}
The authors would like to thank Prof. Max Costa for helpful discussions about the noiseberg region.

\appendix
In the appendix we prove the equivalence between the achievable rate region in the overflow region and $\mathcal{R_{N}}$.

Comparing this region with $\mathcal{R_{N}}$, we see that (\ref{eq:9})(\ref{eq:10}) and (\ref{eq:2})(\ref{eq:3}) are identical. Then it remains to show that (\ref{eq:7})(\ref{eq:11})-(\ref{eq:8}) are redundant given (\ref{eq:9})(\ref{eq:10}).

Eq.~(\ref{eq:7}) is redundant since
\small
\begin{eqnarray}
&& \frac{1}{2}\log \left(1+\frac{P_{1A}}{\bar{\lambda}}\right) +  \frac{1}{2} \log \left(1+ \frac{a^{2} \frac{P_{1C}}{\bar{\lambda}}}{1+a^{2}\frac{P_{1A}}{\bar{\lambda}}+ \frac{P_2}{\bar{\lambda}}}\right) \nonumber \\
&&\leq \frac{1}{2}\log \left(1+\frac{P_{1A}}{\bar{\lambda}}+\frac{P_{1C}}{\bar{\lambda}}\right) \nonumber \\
&\Leftrightarrow& \frac{a^{2} P_{1C}}{\bar{\lambda} + a^{2}P_{1A} + P_2} \left(1 + \frac{P_{1A}}{\bar{\lambda}}\right) \leq \frac{P_{1C}}{\bar{\lambda}} \nonumber \\
&\Leftrightarrow& a^{2}\bar{\lambda} + a^{2}P_{1A} \leq \bar{\lambda} + a^{2}P_{1A} + P_2 \nonumber \\
&\Leftrightarrow& (1-a^{2})\bar{\lambda} + P_2 \geq 0, \nonumber 
\end{eqnarray}
\normalsize
which is trivially true.

Eq.~(\ref{eq:11}) is redundant, since it is superseded by (\ref{eq:9}) and (\ref{eq:10}), i.e., we will show
\small
\begin{eqnarray}
&& \frac{1}{2}\log \left(1+\frac{P_{1A}}{\bar{\lambda}}\right) +  \frac{1}{2} \log \left(1+ \frac{a^{2} \frac{P_{1C}}{\bar{\lambda}}}{1+a^{2}\frac{P_{1A}}{\bar{\lambda}}+ \frac{P_2}{\bar{\lambda}}}\right) \nonumber \\
&&+\frac{1}{2} \log \left(1+ \frac{ \frac{P_2}{\bar{\lambda}}}{1+a^{2}\frac{P_{1A}}{\bar{\lambda}}}\right)  \nonumber\\
&&\leq \frac{1}{2}\log \left(1+\frac{P_{1A}}{\bar{\lambda}}+\frac{P_{1C}}{\bar{\lambda}}+\frac{b^{2}P_2}{\bar{\lambda}}\right),\nonumber
\end{eqnarray}
\normalsize
or, equivalently,
\small
\begin{eqnarray}
&&\frac{a^{2}P_{1C}}{\bar{\lambda}+a^{2}P_{1A}+P_2}\left(1+ \frac{P_{1A}}{\bar{\lambda}}\right)+ \frac{P_2}{\bar{\lambda}+a^{2}P_{1A}}\left(1+ \frac{P_{1A}}{\bar{\lambda}}\right) \nonumber\\
&&+ \frac{P_2}{\bar{\lambda}+a^{2}P_{1A}}\left(1+ \frac{P_{1A}}{\bar{\lambda}}\right)\frac{a^{2}P_{1C}}{\bar{\lambda}+a^{2}P_{1A}+P_2} \nonumber\\
&&\leq  \frac{P_{1C}}{\bar{\lambda}}+\frac{b^{2}P_2}{\bar{\lambda}}. \label{eq:13}
\end{eqnarray}
\normalsize
In order to prove (\ref{eq:13}), it suffices to show
\small
\begin{eqnarray}
\frac{P_2}{\bar{\lambda}+a^{2}P_{1A}}\left(1+ \frac{P_{1A}}{\bar{\lambda}}\right) \leq \frac{b^{2}P_2}{\bar{\lambda}} \label{eq:14}
\end{eqnarray}
\normalsize
and
\small
\begin{eqnarray}
&& \frac{P_2}{\bar{\lambda}+a^{2}P_{1A}}\left(1+ \frac{P_{1A}}{\bar{\lambda}}\right)\frac{a^{2}P_{1C}}{\bar{\lambda}+a^{2}P_{1A}+P_2} \nonumber \\
&&+\frac{a^{2}P_{1C}}{\bar{\lambda}+a^{2}P_{1A}+P_2}\left(1+ \frac{P_{1A}}{\bar{\lambda}}\right) \leq  \frac{P_{1C}}{\bar{\lambda}}. \label{eq:15}
\end{eqnarray}
\normalsize
Eq.~(\ref{eq:14}) is equivalent to
\small
\begin{eqnarray}
&&\frac{b^{2}}{\bar{\lambda}} \geq \frac{1+ \frac{P_{1A}}{\bar{\lambda}}}{\bar{\lambda} + a^{2}P_{1A}} \nonumber \\
&\Leftrightarrow& (b^{2}-1)\bar{\lambda}+(a^{2}b^{2}-1)P_{1A} \geq 0 \nonumber 
\end{eqnarray}
\normalsize
which is obviously true.
Eq.~(\ref{eq:15}) is equivalent to
\small
\begin{eqnarray}
&& \frac{P_{1C}}{\bar{\lambda}} \geq \frac{a^{2}P_{1C}}{\bar{\lambda}+a^{2}P_{1A}+P_2}\left(1+ \frac{P_{1A}}{\bar{\lambda}}\right)\left(1+\frac{P_2}{\bar{\lambda}+a^{2}P_{1A}}\right)\nonumber\\ \label{eq:17}\\
&\Leftrightarrow& \bar{\lambda} + a^{2}P_{1A} + P_2 \geq a^{2}(\bar{\lambda} + P_{1A})\left(1+\frac{P_2}{\bar{\lambda}+a^{2}P_{1A}}\right) \nonumber \\
&\Leftrightarrow& \bar{\lambda}(1-a^{2})(\bar{\lambda}+a^2P_{1A}+P_2) \geq 0. \nonumber 
\end{eqnarray}
\normalsize
Thus (\ref{eq:14}) and (\ref{eq:15}) are true and (\ref{eq:11}) is redundant.

Comparing (\ref{eq:9})(\ref{eq:10}) and (\ref{eq:12}), we see that (\ref{eq:12}) being redundant is equivalent to
\small
\begin{eqnarray}
&& \frac{1}{2}\log \left(1+\frac{P_{1A}}{\bar{\lambda}}\right)+\frac{1}{2} \log \left(1+ \frac{ \frac{P_2}{\bar{\lambda}}}{1+a^{2}\frac{P_{1A}}{\bar{\lambda}}}\right) \nonumber \\
&& \leq \frac{1}{2}\log \left(1+\frac{P_{1A}}{\bar{\lambda}}+\frac{b^{2}P_2}{\bar{\lambda}}\right) \label{eq:16}\\
&\Leftrightarrow& 1+ \frac{P_{1A}}{\bar{\lambda}} + \frac{P_2}{\bar{\lambda}+a^{2}P_{1A}}+ \frac{P_{1A}}{\bar{\lambda}}\frac{P_2}{\bar{\lambda}+a^{2}P_{1A}} \nonumber \\
&& \leq 1 + \frac{P_{1A}}{\bar{\lambda}} +\frac{b^{2}P_2}{\bar{\lambda}} \nonumber
\end{eqnarray}
\normalsize
which is equivalent to (\ref{eq:14}). Thus (\ref{eq:12}) is redundant.

Comparing (\ref{eq:9})(\ref{eq:10}) and (\ref{eq:8}), we see that (\ref{eq:8}) being redundant is equivalent to
\small
\begin{eqnarray}
&& 2*\frac{1}{2} \log \left(1+ \frac{ \frac{P_2}{\bar{\lambda}}}{1+a^{2}\frac{P_{1A}}{\bar{\lambda}}}\right) + \frac{1}{2}\log \left(1+\frac{P_{1A}}{\bar{\lambda}}\right) \nonumber \\ &&+  \frac{1}{2} \log \left(1+ \frac{a^{2} \frac{P_{1C}}{\bar{\lambda}}}{1+a^{2}\frac{P_{1A}}{\bar{\lambda}}+ \frac{P_2}{\bar{\lambda}}}\right) \nonumber \\
&& \leq \frac{1}{2}\log \left(1+\frac{P_{1A}}{\bar{\lambda}}+\frac{b^{2}P_2}{\bar{\lambda}}\right) + \frac{1}{2}\log \left(1+\frac{P_{1C}}{\bar{\lambda}}+\frac{b^{2}P_2}{\bar{\lambda}}\right). \nonumber
\end{eqnarray}
\normalsize
It suffices to show (\ref{eq:16}) and
\small
\begin{eqnarray}
&&\frac{1}{2} \log \left(1+ \frac{a^{2} \frac{P_{1C}}{\bar{\lambda}}}{1+a^{2}\frac{P_{1A}}{\bar{\lambda}}+ \frac{P_2}{\bar{\lambda}}}\right)+\frac{1}{2} \log \left(1+ \frac{ \frac{P_2}{\bar{\lambda}}}{1+a^{2}\frac{P_{1A}}{\bar{\lambda}}}\right)\nonumber\\
&& \leq \frac{1}{2}\log \left(1+\frac{P_{1C}}{\bar{\lambda}}+\frac{b^{2}P_2}{\bar{\lambda}}\right) \nonumber
\end{eqnarray}
\normalsize
which is equivalent to
\small
\begin{eqnarray}
&&1 + \frac{P_2}{\bar{\lambda}+a^{2}P_{1A}}  + \frac{P_2}{\bar{\lambda}+a^{2}P_{1A}}\frac{a^{2}P_{1C}}{\bar{\lambda}+a^{2}P_{1A}+P_2} \nonumber \\
&&+ \frac{a^{2}P_{1C}}{\bar{\lambda}+a^{2}P_{1A}+P_2} \leq 1 + \frac{P_{1C}}{\bar{\lambda}} + \frac{b^{2}P_2}{\bar{\lambda}}. \nonumber
\end{eqnarray}
\normalsize
Then it is sufficient to show
\small
\begin{eqnarray}
\frac{P_2}{\bar{\lambda}+a^{2}P_{1A}} \leq \frac{b^{2}P_2}{\bar{\lambda}} \nonumber \label{eq:18}
\end{eqnarray}
\normalsize
and
\small
\begin{eqnarray}
\left(1+\frac{P_2}{\bar{\lambda}+a^{2}P_{1A}}\right)\frac{a^{2}P_{1C}}{\bar{\lambda}+a^{2}P_{1A}+P_2} \leq \frac{P_{1C}}{\bar{\lambda}}, \nonumber \label{eq:19}
\end{eqnarray}
\normalsize
which are trivially true from (\ref{eq:14}) and (\ref{eq:17}).  Thus (\ref{eq:8}) is redundant given (\ref{eq:9}) and (\ref{eq:10}).

\end{document}